\documentclass[11pt]{article}
\usepackage[T1]{fontenc}
\usepackage[utf8]{inputenc}
\usepackage[margin=1.25in]{geometry}
\usepackage{setspace}
\usepackage{amsmath,amssymb,amsthm,mathtools,bm}
\usepackage{booktabs,array,tabularx,graphicx,enumitem}
\usepackage{tikz,pgfplots}
\pgfplotsset{compat=1.18}
\usetikzlibrary{positioning,arrows.meta}
\usepackage{float}
\floatstyle{ruled}
\newfloat{algorithm}{tbp}{loa}
\floatname{algorithm}{Algorithm}
\usepackage{flafter}
\usepackage[authoryear,round]{natbib}
\usepackage{authblk}

\newtheorem{lemma}{Lemma}
\newtheorem{theorem}{Theorem}
\newtheorem{proposition}{Proposition}
\newtheorem{corollary}{Corollary}
\newcommand{\Halmos}{\qedhere}
\usepackage{xcolor}
\usepackage{hyperref}
\hypersetup{colorlinks=true,linkcolor=blue!70!black,citecolor=blue!70!black,urlcolor=blue!70!black,
pdftitle={On the Reliability Limits of LLM-Based Multi-Agent Planning},
pdfauthor={Ruicheng Ao, Siyang Gao, David Simchi-Levi},pdfcreator={LaTeX},pdfsubject={},pdfkeywords={}}
\ifdefined\pdfinfoomitdate\pdfinfoomitdate=1\fi
\ifdefined\pdftrailerid\pdftrailerid{}\fi
\ifdefined\pdfsuppressptexinfo\pdfsuppressptexinfo=15\fi
\newcommand{\E}{\mathbb E}
\newcommand{\Prob}{\mathbb P}
\newcommand{\ind}{\mathbf 1}
\newcommand{\cA}{\mathcal A}
\newcommand{\cX}{\mathcal X}
\newcommand{\cM}{\mathcal M}
\newcommand{\cF}{\mathcal F}

\newcommand{\bits}{\{0,1\}}

\graphicspath{{figures/}}
\setlist[enumerate]{topsep=4pt,itemsep=3pt,parsep=0pt}

\usepackage{flafter}

\title{\Large On the Reliability Limits of LLM-Based Multi-Agent Planning}
\author[1]{Ruicheng Ao\thanks{E-mail: \texttt{aorc@mit.edu}, \texttt{siyangao@cityu.edu.hk}, \texttt{dslevi@mit.edu}.}}
\author[2]{Siyang Gao}
\author[3]{David Simchi-Levi}
\affil[1]{Institute for Data, Systems, and Society, Massachusetts Institute of Technology}
\affil[2]{Department of Systems Engineering, City University of Hong Kong}
\affil[3]{Institute for Data, Systems, and Society, Department of Civil and Environmental Engineering, and Operations Research Center, Massachusetts Institute of Technology}
\date{October 31, 2026}
\begin{document}
\maketitle
\begin{abstract}
\noindent
LLM-based multi-agent planning delegates operational decisions to specialists that communicate through finite reports. In this research, we study its reliability limits when a coordinator can revisit source records before completing a resource-constrained plan. A useful report can support an immediate decision or direct later reading, so its value depends on how these stages communicate. We derive exact communication and access tradeoffs and prove that two report forms suffice for a broad class of allocation requirements, one favoring immediate resource identification and the other distinguishing absence from unnamed availability. For shared capacity, this reduction permits joint optimization of reporting roles and adaptive reads for a fixed number of source types with unequal workloads and reading costs. Identical specialists can optimally report different facts, and maximizing immediate completion can impose an unbounded recovery cost penalty even with full-record access. We also conduct numerical experiments to quantify the benefits of the proposed joint optimization and show how source access improves language-agent decisions that already use their reports optimally. In particular, an exhaustive evaluation of 207,632 report designs on small allocation instances recovers the predicted joint optima. Our analysis suggests that reliable agent communication should be evaluated by the complete decisions it supports and the cost of finishing its unresolved cases, instead of solely by report length, local prediction accuracy, or immediate task completion, as is often the case in current practice.
\end{abstract}

\section{Introduction}
\label{sec:introduction}

LLM-based multi-agent planning is more and more used to organize language-agent applications that require decomposition, tool use, verification, and exception handling. Systems for supply-chain optimization, software development, and travel planning distribute these activities among agents that process records and coordinate actions \citep{wu2023autogen,hong2024metagpt,bui2026himap}. Such systems connect language interfaces to operational decisions whose validity depends on shared budgets, capacities, and service requirements. A failure to combine locally useful information can therefore invalidate the complete plan. These systems thus delegate operational decisions across stages that exchange local information.

Within these systems, planning and review can involve different calculations using overlapping records. A critic may detect a capacity violation in a proposed schedule without opening another record. A specialist may instead retrieve an omitted provider and make a previously unsupported assignment possible. Both activities can improve the final decision, but they change different parts of the problem. Additional computation can improve the use of current information, and additional source access can change which decisions that information supports. For this reason, a quantitative account of reliability must relate these operations to the content retained during communication.

Consider a service center with regular capacity for eleven of twelve appointments. Each appointment may also have a compatible external provider, whose identity appears in its own record. One identified external provider, together with regular capacity, is sufficient to serve all twelve appointments. A specialist can report whether a specified provider is available, or whether any provider is available. The first report sometimes completes the plan immediately. When positive, the second identifies a record whose later reading will complete the plan. These equally short and accurate reports leave different work for the coordinator. A report's value therefore includes both the decisions it supports now and the effort it leaves for later.

This issue bears directly on current discussions of LLM agents. Comparisons of architectures find that gains from additional agents depend strongly on the task and coordination structure \citep{kim2025scaling}. Studies of failed executions identify coordination and verification failures alongside local errors \citep{cemri2025multiagent}. Self-correction can be difficult without external feedback \citep{huang2024selfcorrect}, while checking key conditions and using external tools can improve performance \citep{gou2024critic}. These findings motivate a more precise account of what a given collection of reports permits a planner to accomplish, and how much further source access is required when those reports are insufficient.

To understand these findings, we examine the planning episode that produces the final answer. Final accuracy alone does not identify which facts reports discard, whether further computation can recover a valid plan, or how much new source access is necessary. Existing work on comparing information describes what decisions fixed evidence supports. To evaluate a reporting design here, we must also determine the cost of recovering omitted facts. A planner may face many feasible assignments, sources with different reading costs, and reports selected for immediate usefulness. Establishing feasibility can require different information from identifying the resources that implement a plan. The resulting optimization problem couples what specialists communicate initially with which sources the coordinator subsequently reads.

We study this problem through finite models in which specialists report local source states and a coordinator can revisit sources before completing a plan. It may use any available option, including unreported ones, and must recognize infeasibility when no valid plan exists.

We first consider synchronized resources, where two providers must offer the same service slot, and derive the exact tradeoff between report capacity and later checks. Reports with identical lengths and identical local accuracy can leave substantially different recovery costs. Our main structural result covers allocation requirements that remain feasible when additional task-specific resources become available. Independent sources can differ in their probabilities, report capacities, and reading costs. When resource identities within a source have the same operational contribution, we prove that two report forms suffice for global optimality. A report partitions local states into groups that produce the same message. The result replaces the search over all such partitions by a choice between two reporting roles at each source. For shared capacity and a fixed number of source types, we then optimize those roles together with adaptive reads under unequal task workloads and reading costs. The exact calculation uses counts within each type, and its optimum is over all permitted local report partitions.

For identical sources and unit workloads, this joint solution gives an explicit access rule and report allocation. Specialists with the same information capacity can optimally use different reporting instructions. Further results show that optimizing immediate performance can leave expensive recovery, including an unbounded ratio under complete-source access. These comparisons identify how the operational requirement to finish a plan changes the preferred content of local reports.

We also conduct numerical experiments to examine these reporting choices. We first use exact calculations to test the theoretical characterizations and compare assignment costs across reporting rules. Language agent experiments then implement the same tasks and measure complete-assignment correctness and source reading. A further comparison examines when reviewing a proposed plan using existing reports suffices and when returning to source records improves the decision.

\textit{Related research.} This research is related to three streams of literature. The first concerns information and coordinated decisions. Existing research on team decision theory studies how shared and private observations affect coordinated actions \citep{radner1962team}. There is also a literature on prescriptive analytics that evaluates information through the decisions it supports \citep{bertsimas2020prescriptive,elmachtoub2022smart}. Our setting applies this perspective to reports followed by costly source access.

The second stream studies communication and information acquisition for downstream decisions. Research on functional source coding examines what messages must preserve \citep{doshi2010functional}, and work on action-dependent side information connects those messages to later information acquisition \citep{permuter2011vending}. Related research on adaptive testing selects observations until a decision can be made \citep{javdani2014decision}. For allocation problems, the communication-constrained matching literature studies what each source should transmit to preserve a global assignment \citep{ahmadian2026sparsification,mohanty2026matching}. In our model, establishing that resources exist and naming them for execution are separate requirements. We jointly choose reports and later reads, using the two-report reduction to optimize over all local partitions.

The third stream develops and evaluates LLM-based planning systems. Existing frameworks such as AutoGen and MetaGPT combine decomposition, communication, and tool use \citep{wu2023autogen,hong2024metagpt}. Research on complete-plan evaluation and formal tools examines feasibility under operational constraints \citep{xie2024travelplanner,hao2025formal}. There is also work on observation compression that evaluates retained information through its effect on agent actions \citep{chen2026coact}. Our analysis connects these choices to the information required for reliable completion and the cost of recovering it.

\section{Delegated Planning with Reports and Source Access}
\label{sec:model}

A typical LLM-based planning system receives a request, distributes work among specialists, and combines their reports into a final action. Specialists share the task instructions but may inspect records unavailable to the coordinator. For example, a service specialist reads the providers compatible with an appointment and reports selected availability information. The coordinator combines these reports with the capacity requirements and can request further records when necessary.

The service center example in Section~\ref{sec:introduction} identifies the information used at each stage. A source is a local record available to a specialist, such as the record describing the external providers compatible with one appointment. A report is a capacity-constrained message derived from that record, and source access is a later request for additional information from it. The common context contains the appointment list and regular capacity. More generally, tasks can require different amounts of capacity, and records can have different retrieval costs. Checking the capacity of a proposed assignment uses the reports already received. Reading an omitted provider's record supplies additional information. We first define these elements and then explain how further calculation and source access affect reliability.

\subsection{Sources, Reports, and Complete Plans}

The common context contains information shared by all agents, such as the task
list and regular capacity, and the request specifies the particular
planning problem to be answered. The local sources contain information that is
initially available only to individual specialists. Let $X=(X_1,\ldots,X_L)$ be a finite source state, let $B$ denote the common context, and let $W\in\mathcal W$ be a finite-valued planning request. Specialist $i$ observes $(X_i,B)$ with $X_i\in\cX_i$. If the request is known before reporting, its message is $M_i=f_i(X_i,B,W)$ in an alphabet $\cM_i$ with $|\cM_i|\le b_i$. Otherwise the report depends on $(X_i,B)$, and the coordinator receives $W$ later. The report maps are selected jointly before the local states are observed. We condition on the realized common context and suppress $B$ below, and suppress $W$ when it is fixed. Section~\ref{sec:selection} will consider requests received after reporting.

Write $M=(M_1,\ldots,M_L)$, $f=(f_1,\ldots,f_L)$, and $b=(b_1,\ldots,b_L)$. Capacity $b_i$ counts possible reports, so a fixed binary representation uses $\lceil\log_2b_i\rceil$ bits. Each message groups together local source states that produce that message. We call this group a report cell, and the collection of cells a report partition. A language implementation specifies the report's fields and their meanings. Its actual message length determines the communication cost. For example, one binary field can identify whether a particular provider is available. Another can identify whether any provider is available. Both fields have two possible values, but they support different subsequent decisions.

Reliability means making the correct decision for the full request. Let the finite set $\cA$ contain all candidate complete plans and the infeasibility response. For source state $x$ and request $w$, let $\cF(x,w)$ contain the feasible complete plans. With $\bot$ denoting infeasibility, the set of correct actions is
\[
 \cA^*(x,w)=
 \begin{cases}
 \cF(x,w),&\cF(x,w)\ne\varnothing,\\
 \{\bot\},&\cF(x,w)=\varnothing.
 \end{cases}
\]
The loss is one if the selected action is not in this set. Under this definition, an invalid assignment and an incorrect infeasibility response are both failures. Reports restrict observations, and every original action remains available, including an assignment using an unreported provider.

When the initial reports do not determine a correct final action, the
coordinator may return to the underlying records. We refer to each such later
information request as a source-access operation, or query. A finite set $\mathcal T$ specifies the available operations. Operation $t\in\mathcal T$ returns $g_t(X)$ at cost $c_t>0$. Depending on the source interface, it can check one named record, retrieve a filtered response, or read a complete local source. A policy $\gamma$ selects operations using $(M,W)$ and the preceding responses, then returns $A_\gamma$. Let $Q_\gamma$ count these operations and let $C_\gamma$ be their total cost. We refer to this process of acquiring additional source information to complete an unresolved decision as recovery. The failure probability is $R(f,\gamma)=\Prob\{A_\gamma\notin\cA^*(X,W)\}$. Guaranteed correctness means correctness with probability one over internal randomization for every admissible source state and request. When every admissible case has positive probability, this is equivalent to zero risk. We evaluate a fixed report rule through minimum risk under a query cap, and choose reports jointly with queries through minimum expected recovery cost,
\[
 R_q^*(f)=\inf_{\gamma:Q_\gamma\le q}R(f,\gamma),
 \qquad C^*(b)=\min_{f,\gamma:\,\text{guaranteed correct}}\E C_\gamma.
\]
We set $C^*(b)=+\infty$ when guaranteed completion is impossible. We also
consider the worst-case number of queries required for guaranteed completion,
which captures the hard completion requirement. The optimization concerns source
access after the initial inspection and reporting stages. If their implementation cost is $C_0$, the total cost is $C_0+C_\gamma$. A common $C_0$ preserves the ranking of reporting designs, while its size affects their relative total savings.

\subsection{Information and Computation in Recovery}

To distinguish these two ways of improving a decision, let $H$ contain the common context, request, reports, and source responses available at a given stage. A coordinator can calculate consequences of $H$ or acquire another source response. For example, it can count assigned appointments using its current reports, or ask a specialist for an omitted provider's identity. A specialist that retained its complete record can supply that identity directly. This response is source access from the coordinator's perspective.

For a known finite law, the minimum risk based on $H$, averaged over histories, is
\[
 V(H)=\E\left[1-\max_{a\in\cA}\Prob\{a\in\cA^*(X,W)\mid H\}\right].
\]
This benchmark chooses an action with the greatest conditional probability of being correct. To evaluate an implemented planner, write its risk as $V(H)+e(H)$, where $e(H)\ge0$ is the excess risk from its use of the available information. A reviewer may revise a proposed plan or check its feasibility against known constraints. Such calculation can reduce $e(H)$. Acquiring a new source response can also reduce the minimum achievable risk $V(H)$. Thus the improvement available from reviewing the same reports depends on how well the planner already uses them.

A corresponding condition determines when guaranteed completion is possible without another read. Let $\mathcal S(H)$ contain every admissible source state consistent with $H$, regardless of its probability. The coordinator can stop with a guaranteed correct action exactly when
\[
 \bigcap_{x\in\mathcal S(H)}\cA^*(x,W)\ne\varnothing.
\]
The action must remain valid in every possible state, although each state may admit several plans. In the service example, enough identified providers cover the shortage even when other records remain unread.

Before introducing the models below, we state their common setting. Source states remain unchanged during a planning episode, and reports and query responses are accurate. This setting represents planning based on a given version of a scheduling or resource database. Later numerical experiments use natural-language records and measure both reporting and final-decision errors against the same complete source state.

\section{Reports and Checks for Synchronized Assignments}
\label{sec:query}

We first consider a synchronized assignment, in which two providers must provide the same service slot. This requirement lets us determine the recovery cost after every pair of reports. The first source records its provider's available slot $X\in[n]$, and the second records $Y\in[n]$, with all $n^2$ pairs admissible. A feasible assignment uses label $j$ when $X=Y=j$. Otherwise the correct answer is $\bot$. A named-record check asks whether $X=j$ or $Y=j$ and returns a binary answer at unit cost. This interface describes checking a proposed option in a source record. We determine the checks required after a given pair of reports and then optimize the reports themselves in this section. Section~\ref{sec:capacity-assignments} will consider complete-source reads.

\subsection{Recovery from a Pair of Reports}

Each received report leaves a set of possible local labels. For report maps $f,g$ and received symbols $u,v$, let these sets be $A=f^{-1}(u)$ and $B=g^{-1}(v)$. Their remaining joint states form $A\times B$. Write $H=A\cap B$ and $h=|H|$. Labels in $H$ can still form a complete assignment, while labels in $A\setminus B$ or $B\setminus A$ cannot. Let $D(A,B)$ be the smallest worst-case number of adaptive checks needed for correctness throughout the cell.

\begin{lemma}
\label{thm:cell}
For nonempty $A,B\subseteq[n]$,
\[
D(A,B)=
\begin{cases}
0,&h=0\text{ or }A=B=\{j\}\text{ for some }j,\\
h,&h\ge1,\ A\subseteq B\text{ or }B\subseteq A,\quad |A\cup B|>1,\\
h+1,&h\ge1,\ A\setminus B\ne\varnothing,\ B\setminus A\ne\varnothing.
\end{cases}
\]
Randomization cannot reduce this number when correctness and the cap hold almost surely for every state.
\end{lemma}

In the last case, each report permits labels excluded by the other. The extra check distinguishes these exclusive labels from the common candidates. For example, two reports that both permit label $j$ do not establish compatibility when each also permits an exclusive label. If both reports identify $j$ uniquely, the assignment is already determined.

The same effect appears in the risk under a query cap. Conditional on equal source states uniformly distributed over a block of $h$ labels, any rule using at most $q$ checks has success probability at most $\min\{1,(q+1)/h\}$.
Along a sequence of negative answers, the rule tests at most $q$ distinct labels, and its final action can identify at most one more. This permits arbitrary calculation and a final guess. Extra reasoning can improve an implemented rule toward this limit, while changing the limit requires a different report or additional source access.

\subsection{Optimal Report Capacity and Recovery Cost}
\label{sec:design}

Suppose the reports use at most $a$ and $b$ symbols. Let $N_q(a,b)$ be the largest number of resource labels supported with guaranteed correctness after at most $q$ checks. For report cells $A_i,B_j$, set $h_{ij}=|A_i\cap B_j|$ and $E=\{(i,j):h_{ij}>0\}$. Each source label contributes to exactly one intersection, so $n=\sum_{(i,j)\in E}h_{ij}$. To represent these intersections, construct a bipartite graph with one vertex per nonempty report cell and an edge for each $(i,j)\in E$. In this support graph, a cell is contained in its neighboring cell exactly when its vertex has degree one. Therefore, for $q\ge1$, Lemma~\ref{thm:cell} gives $h_{ij}\le q-1+\ind\{\deg(i)=1\text{ or }\deg(j)=1\}$.
Conversely, any positive integer weights satisfying these restrictions define reports by assigning one label to each unit of edge weight. Thus the choice of reports becomes an exact weighted-graph optimization.

\begin{theorem}
\label{thm:capacity}
For positive integers $a,b$ and integer query caps $q\ge0$,
\begin{align}
 N_0(a,b)&=\min\{a,b\},\nonumber\\
 N_1(a,b)&=
 \begin{cases}\max\{a,b\},&\min\{a,b\}=1,\\a+b-2,&a,b\ge2,\end{cases}\nonumber\\
 N_q(a,b)&=
 \begin{cases}q\max\{a,b\},&\min\{a,b\}=1,\\
 \max\{(q-1)ab,\ q(a+b-1)-1\},&a,b\ge2,
 \end{cases}\quad q\ge2.
 \nonumber
\end{align}
Then the smallest worst-case query cap is $Q^*(n;a,b)=\min\{q:n\le N_q(a,b)\}$. Deterministic reports attain every stated capacity.
\end{theorem}

For $a,b\ge2$ and $q\ge2$, the first capacity term is attained by complete bipartite support with $q-1$ labels per intersection. The second uses two opposite-side centers, with $q-1$ labels on their connecting edge and $q$ labels on each edge joining a center to an opposite leaf. Its leaf cells save one subsequent check. For example, two symbols per specialist and two checks support five labels through the latter support, compared with four through complete support. The reports can be stored as two label-to-cell maps and the occupied intersection lists, using $O(n+a+b)$ space and setup operations.

This characterization also lets us compare reports that retain equal amounts of information about each source. For integers $r\ge3$ and $s\ge1$, identify the $n=r^2s$ labels with $[r]\times[r]\times[s]$. Under \emph{identical reports}, both specialists send the first coordinate. Under \emph{crossed reports}, the first sends the first coordinate and the second sends the second coordinate. Each report has $r$ cells of size $rs$. Thus each specialist communicates the same amount of local information in both designs, although the pairs of reports retain different facts about compatibility.

\begin{proposition}
\label{prop:balanced}
Let $n=r^2s$, $r\ge3$, $s\ge1$, and $h=rs$. Identical reports have optimal worst-case query cost $h$. Crossed reports have optimal worst-case cost $s+1$, which equals $Q^*(n;r,r)$. Under uniform source marginals, each design has report entropy $\log r$, conditional source entropy $\log h$, and optimal local reconstruction probability $1/h$.

For $\pi\in[0,1]$, let $\Prob(X=x,Y=y)=\pi/n$ when $x=y$ and $(1-\pi)/[n(n-1)]$ otherwise. Every rule that is correct on all source states after identical reports satisfies
\[
 \E Q\ge\pi\frac{(h-1)(h+2)}{2h}.
\]
Every rule with at most $q$ checks has failure probability at least $\pi[1-(q+1)/h]_+$. Crossed reports permit correctness on every state with at most $s+1$ checks.
\end{proposition}

Proposition~\ref{prop:balanced} shows that local information measures can agree even when the cost of reliable joint planning differs by a factor approaching $r$. The reason is visible in the report intersections. Identical cells leave $rs$ compatible labels, but crossed cells leave only $s$. The average gain also depends on $\pi$. When equality is rare, many pairs of identical reports establish infeasibility immediately. Consequently, the relevant report comparison must account for both the dependence between source records and the planning objective.

\section{Joint Reporting and Reliable Resource Allocation}
\label{sec:capacity-assignments}

This section studies how to jointly choose reports and subsequent source access. We first reduce the admissible report forms, then solve the resulting problem with several source types and examine its implications for identical sources.

\subsection{Local Resources and Feasible Plans}

A specialist's report can identify a resource for immediate use or indicate where further reading will obtain one. To compare these choices, we allow sources to have different distributions, report capacities, and reading costs while keeping their operational contributions explicit.

Source $i\in[n]$ has state $X_i\in\{0,1,\ldots,m_i\}$. The states are independent, with $\Prob(X_i=j)=p_{ij}>0$. State zero means absence. State $j\ge1$ identifies resource $(i,j)$. The positive states of a source have the same operational contribution, but the coordinator must name the actual resource when using it. For example, an outside provider can serve a particular appointment, and its identity determines whom to contact. Provider identities at different sources are different. Reading source $i$ returns its full record at cost $c_i>0$.

To describe which resource collections complete the task, let $\mathcal F\subseteq2^{[n]}$ satisfy $S\in\mathcal F$ and $S\subseteq T\subseteq[n]\Rightarrow T\in\mathcal F$. Such a family is called upward closed. A plan specifies a set $S\in\mathcal F$ and the available resource name at each $i\in S$. A capacity requirement uses $\mathcal F=\{S:\sum_{i\in S}w_i\ge d\}$, where $w_i$ is the workload served by resource $i$ and $d$ is the capacity shortage. A coverage requirement can instead require the chosen sources to cover every specified service. Extra available sources can be left unused, which explains upward closure. If no feasible set of available sources exists, the correct response is $\bot$.

Specialist $i$ has at most $b_i$ possible report symbols, where $2\le b_i\le m_i$. Its reporting rule is selected before its source state is observed. Each reporting rule groups states that send the same symbol. Let $\mathcal P_i$ contain all such partitions of $\{0,\ldots,m_i\}$ into at most $b_i$ nonempty cells. For $P=(P_i)_{i=1}^n\in\prod_i\mathcal P_i$, let $\Gamma(P)$ contain the deterministic adaptive reading rules that complete every source state correctly, and let $R_\gamma$ be the set of sources read by $\gamma$. The expected recovery cost of $P$ is $C(P)=\min_{\gamma\in\Gamma(P)}\E[\sum_{i\in R_\gamma}c_i]$.
Reading every source provides one such rule. Repeated reads can be omitted, so the minimum is over a finite collection of decision rules.

We now define two report forms that preserve different facts for later allocation. Relabel positive states so that $p_{i1}\ge\cdots\ge p_{im_i}$ and set
\begin{align*}
 L_i&=\bigl\{\{1\},\ldots,\{b_i-1\},\{0,b_i,\ldots,m_i\}\bigr\},\\
 S_i&=\bigl\{\{0\},\{1\},\ldots,\{b_i-2\},\{b_i-1,\ldots,m_i\}\bigr\}.
\end{align*}
The \emph{location report} $L_i$ names $b_i-1$ positive states individually. The \emph{status report} $S_i$ identifies absence and names $b_i-2$ positive states individually. Each uses $b_i$ symbols, with an empty singleton list omitted. When $b_i=2$, $S_i$ reports only absence or availability.

\begin{theorem}
\label{thm:report-reduction}
For every upward-closed $\mathcal F\subseteq2^{[n]}$, independent source laws with $p_{ij}>0$, reading costs $c_i>0$, and report budgets $2\le b_i\le m_i$,
\[
 \min_{P\in\prod_{i=1}^n\mathcal P_i}C(P)
 =\min_{P\in\prod_{i=1}^n\{L_i,S_i\}}C(P).
\]
Thus there exist reporting rules $P_i^*\in\{L_i,S_i\}$ and a continuation $\gamma^*\in\Gamma(P^*)$ attaining the unrestricted joint optimum. Shared or private randomization independent of the sources does not reduce this value.
\end{theorem}

Theorem \ref{thm:report-reduction} shows that the apparent complexity of local report design is unnecessary for this class of planning problems. Although a source may use any admissible partition, a globally optimal system can always be constructed by assigning each source one of only two operationally distinct reporting roles. To understand why these two forms suffice, consider a report that groups several positive states. It establishes that its source can contribute, but its resource name must still be read. Grouping several such messages changes no subsequent allocation choice. When one of these groups exists, separating absence from the remaining positive states improves the information available for later decisions. These operations leave the two stated report forms. Naming a more probable positive state also reduces the chance that its identity will require a later read. The choice between the two forms therefore remains a joint decision among sources, even after the individual cells have been specified.

\subsection{Different Source Types and Reading Costs}

Consider tasks with total workload $W=\sum_{t=1}^K n_tw_t$, served by a common pool of capacity $W-d$, where $0<d\le W$. There are $n_t$ tasks of type $t$, each with workload $w_t>0$. A private resource serves its entire task, so an allocation exists exactly when the available private resources cover at least $d$ workload. Within type $t$, the source law is $(p_{tj})_{j=0}^{m_t}$, the report capacity is $b_t$, and the reading cost is $c_t$. Types can represent different appointment durations, provider markets, and source-access tools. The number of types $K$ is fixed for the complexity statement below.

To compute the probabilities of the possible reports, set $\alpha_t=\sum_{j=1}^{b_t-1}p_{tj}$, $\beta_t=\sum_{j=1}^{b_t-2}p_{tj}$, $\rho_t=1-\alpha_t$, $\sigma_t=1-p_{t0}-\beta_t$, and $q_t=(\rho_t-p_{t0})/\rho_t$. An empty sum is zero. Here $\alpha_t,\beta_t$ are the naming probabilities under the two forms, and $\rho_t,\sigma_t$ are their grouped-report probabilities. Conditional on a grouped location report, availability has probability $q_t$. A grouped status report establishes availability without a resource name.

After receiving reports and any later reads, let $a_t,u_t,r_t$ count uncertain sources, unnamed but available resources, and named resources of type $t$. Write $n=(n_1,\ldots,n_K)$ and similarly $a,u,r$ for the count vectors. Let $V(a,u,r)$ be minimum expected remaining reading cost. It is zero if $\sum_tw_tr_t\ge d$ or $\sum_tw_t(r_t+a_t+u_t)<d$. The first case permits an allocation, and the second establishes infeasibility. Let $e_t$ be the $t$th unit vector.

\begin{theorem}
\label{thm:capacity-joint}
At every other state, the optimal continuation value satisfies
\begin{align}
V(a,u,r)=\min\biggl\{&\min_{t:a_t>0}\Bigl[c_t+q_tV(a-e_t,u,r+e_t)
 +(1-q_t)V(a-e_t,u,r)\Bigr],\nonumber\\
&\min_{t:u_t>0}\Bigl[c_t+V(a,u-e_t,r+e_t)\Bigr]\biggr\}.
\label{eq:type-value}
\end{align}
An empty inner minimization is omitted. Choose $h_t$ status reporters of each type, and let $j_t$ count those that name a resource. For $0\le a_t\le n_t-h_t$, $0\le j_t\le h_t$, and $0\le u_t\le h_t-j_t$, define
\[
G_t(a_t,j_t,u_t;h_t)=
\binom{n_t-h_t}{a_t}\rho_t^{a_t}\alpha_t^{n_t-h_t-a_t}
\frac{h_t!\,\beta_t^{j_t}\sigma_t^{u_t}p_{t0}^{h_t-j_t-u_t}}
{j_t!u_t!(h_t-j_t-u_t)!}.
\]
Using $0^0=1$, its minimum expected cost is
\begin{equation}
 C(h)=\sum_{a,j,u}\left[\prod_{t=1}^K G_t(a_t,j_t,u_t;h_t)\right]
 V(a,u,n-h-a+j),
 \label{eq:capacity-joint-cost}
\end{equation}
where $h=(h_1,\ldots,h_K)$, $j=(j_1,\ldots,j_K)$, vectors are added componentwise, and the sum uses the stated ranges. The unrestricted joint optimum is $\min_{0\le h_t\le n_t}C(h)$. It can be computed in
\[
O\left(\sum_{t=1}^Km_t\log m_t+K\prod_{t=1}^K(n_t+1)^4\right)
\]
arithmetic operations and $O\bigl(\sum_tm_t+\prod_t(n_t+1)^3\bigr)$ space. The state space does not depend on the numerical magnitudes of the workloads.
\end{theorem}

Theorem \ref{thm:capacity-joint} provides an exact bridge from local communication design to operational planning performance. It shows how reporting and retrieval decisions can be optimized jointly under unequal workloads and reading costs. Here $n-h-a+j$ counts all named resources. The product of $G_t$ gives the probability of these report counts. Theorem~\ref{thm:report-reduction} makes this optimization global over all admissible local reports. Since sources within each type have identical laws and costs, their reporting assignments reduce to the count vector $h$. Equation~\eqref{eq:type-value} compares the available reads conditional on those reports. Existing research on sequential testing determines a system's condition from costly component observations \citep{BenDov1981}. In the present allocation problem, the information required for execution also includes resource identities. An unread available resource may establish feasibility while still requiring a paid read before the coordinator can assign it. The joint design chooses the initial partitions as well as these later observations.

Reading costs can still affect the choice even when feasibility is already known. Suppose either of two sources can cover the remaining shortage. A source known to be available costs $10$ to read, and an uncertain source has availability probability $0.9$ and reading cost $1$. Reading the certain source costs $10$. Reading the uncertain source first and the certain source only after absence costs $1+0.1(10)=2$ in expectation. Thus further availability checks can be optimal even after feasibility is established.

\begin{algorithm}[tbp]
\caption{Joint reporting and source access with several task types}
\label{alg:capacity-joint}
\textbf{Input} $(n_t,w_t,c_t,b_t,(p_{tj})_{j=0}^{m_t})_{t=1}^K$, shortage $d$.\quad
\textbf{Output} a complete allocation or $\bot$.
\begin{enumerate}[leftmargin=*,label=\arabic*.]
\item Sort each positive-state law. Evaluate $V(a,u,r)$ in increasing $\sum_t(a_t+u_t)$ using~\eqref{eq:type-value} and its zero-cost stopping states.
\item Compute $h^*\in\arg\min_h C(h)$ by~\eqref{eq:capacity-joint-cost}. For each type $t$, assign status reports to $h_t^*$ sources and location reports to the rest.
\item Receive the reports. Form sets $A_t,U_t$ of uncertain and unnamed available sources, and the set $N_t$ of named resources. Set $(a_t,u_t,r_t)=(|A_t|,|U_t|,|N_t|)$ and retain every reported resource identity.
\item Set $\mathcal D(a,u)=\{(t,L):a_t>0\}\cup\{(t,S):u_t>0\}$, and let $Q_t^L,Q_t^S$ be the corresponding action values in~\eqref{eq:type-value}. While $\sum_tw_tr_t<d\le\sum_tw_t(r_t+a_t+u_t)$, choose
\[
(t,z)\in\arg\min_{(s,y)\in\mathcal D(a,u)}Q_s^y(a,u,r).
\]
If $z=L$, read $i\in A_t$, set $A_t\leftarrow A_t\setminus\{i\}$, and, when $X_i>0$, set $N_t\leftarrow N_t\cup\{(i,X_i)\}$.
If $z=S$, read $i\in U_t$ and set $U_t\leftarrow U_t\setminus\{i\}$ and $N_t\leftarrow N_t\cup\{(i,X_i)\}$.
Reset $(a_t,u_t,r_t)=(|A_t|,|U_t|,|N_t|)$ and update $\mathcal D(a,u)$.
\item If $\sum_tw_tr_t\ge d$, use the named private resources and assign all remaining tasks to the common pool. Otherwise return $\bot$.
\end{enumerate}
\end{algorithm}

\subsection{Explicit Rules for Identical Sources}
\label{subsec:pool-one}

With one source type, unit workloads, unit reading cost, and integer shortage $d\in\{1,\ldots,n\}$, the optimum has a simpler form. Suppress type subscripts and write $v=d-r$ for the remaining number of private resources required. To determine whether the uncertain sources can supply $\ell$ resources, define $B_{a,\ell}=\Prob\{\operatorname{Bin}(a,q)\ge\ell\}$. Let $T_{a,\ell}$ be the expected reads needed to decide whether this threshold is met. For $1\le\ell\le a$, the recursions are $B_{a,\ell}=qB_{a-1,\ell-1}+(1-q)B_{a-1,\ell}$ and $T_{a,\ell}=1+qT_{a-1,\ell-1}+(1-q)T_{a-1,\ell}$.
For $\ell\le0$, set $B_{a,\ell}=1$ and $T_{a,\ell}=0$. For $\ell>a$, both are zero.

\begin{corollary}
\label{cor:homogeneous-access}
For identical sources with unit workloads and reading costs and $d\in\{1,\ldots,n\}$,
\begin{equation}
 V(a,u,r)=
 \begin{cases}
 0,&v\le0\text{ or }v>a+u,\\
 v,&1\le v\le u,\\
 T_{a,v-u}+uB_{a,v-u},&u<v\le a+u.
 \end{cases}
 \label{eq:capacity-value}
\end{equation}
The joint optimum can be computed in $O(m\log m+n^4)$ operations and $O(m+n^2)$ space. For $1\le v\le u$, read $v$ unnamed available resources. For $u<v\le a+u$, read uncertain sources until feasibility is determined, then read the $u$ unnamed available resources only if feasible.
\end{corollary}

A one-unit shortage gives an explicit report allocation. Let $F_h$ be minimum expected reads with $h$ status reporters. Set $t=1-\alpha$, $\omega=1-\beta$, and $\xi=\alpha-\beta=p_{b-1}$.

\begin{corollary}
\label{cor:pool-one}
For $d=1$ and $h=0,\ldots,n$,
\begin{equation}
 F_h=(\omega^h-p_0^h)t^{n-h}
 +p_0^h\frac{t(t^{n-h}-p_0^{n-h})}{t-p_0}.
 \label{eq:pool-one-cost}
\end{equation}
Moreover, $F_{h+1}-F_h=t^{n-h-1}(\xi\omega^h-p_0^{h+1})$ for $h<n$. The smallest optimal number of status reporters is
\begin{equation}
 h^*=\min\left\{n,\max\left\{0,
 \left\lceil\frac{\log(p_0/\xi)}{\log(\omega/p_0)}\right\rceil
 \right\}\right\}.
 \label{eq:pool-one-optimum}
\end{equation}
If $\xi\omega^{h^*}=p_0^{h^*+1}$ and $h^*<n$, then $h^*+1$ is also optimal. Otherwise the optimum is unique.
\end{corollary}

The threshold balances immediate naming against locating a source for later completion. With binary reports, $p_0=0.4$, and $p_1=0.1$, it gives $h^*=2$ for every $n\ge2$. Two specialists report availability and the others report a particular name. Thus identical specialists can optimally receive different instructions, and the number assigned to availability reporting can remain fixed as the team grows.

\subsection{Immediate Completion and Recovery Cost}
\label{subsec:pool-separation}

We now compare joint design with selecting reports to maximize the probability of completion before any source read. The procedure remains correct on unresolved cases by continuing to read. Thus this criterion evaluates when a correct procedure finishes. The immediate-error criterion in Section~\ref{sec:selection} instead evaluates an answer required without further access.

\begin{corollary}
\label{cor:pool-separation}
For each integer $k\ge2$, take $n=k^2$, unit workloads and reading costs, $d=1$, and binary reports. Independently, each task has no private resource with probability $1-1/k$, and each of its $k^2$ possible resource names occurs with probability $1/k^3$. Every deterministic report design maximizing correct completion before any source read uses a positive singleton at each source. After any such design, minimum expected complete-source reads equal
\begin{equation}
\frac{(1-k^{-3})\bigl[(1-k^{-3})^{k^2}-(1-k^{-1})^{k^2}\bigr]}
{k^{-1}-k^{-3}}=\Theta(k),
\label{eq:pool-separation}
\end{equation}
and the minimum worst-case count is $k^2$. Status reports at every source require at most one complete read and have expected cost $1-(1-1/k)^{k^2}<1$. The expected recovery-cost ratio is therefore $\Omega(\sqrt n)$.
\end{corollary}

Corollary \ref{cor:pool-separation} shows that maximizing immediate completion can be asymptotically incompatible with inexpensive reliable recovery (rather than just creating a finite-instance tradeoff). The separation uses full source records and persists with an expected $k$ available private resources. Hence neither many feasible alternatives nor a complete-record interface ensures inexpensive recovery after reports selected for immediate completion. The joint objective values both the plans a report completes immediately and the cost of finishing its remaining cases.

\section{Immediate Success and Subsequent Recovery}
\label{sec:selection}

The preceding models evaluate reports through the cost of completing a joint decision. A common alternative is to choose reports by their success before any further source access. We first analyze dedicated-resource and backup requests, and then discuss the implications for report evaluation. This comparison is especially relevant when common requests can be answered directly and less frequent requests require searching for a usable resource.

\subsection{Dedicated Resources and a Shared Backup}

Consider $k\ge1$ dedicated resources and one shared backup, with $N=2^k$ possible backup locations. The availability of dedicated resource $i$ is $C_i\in\bits$, where $C=(C_1,\ldots,C_k)$ is uniform on $\bits^k$. The backup is available at exactly one location. Its location $J$ is uniform on $[N]$ and independent of $C$. An initial source inspection observes $(C,J)$ and produces a report with at most $N$ messages, or $k$ bits. After reporting, independently of $(C,J)$, request $i\in[k]$ occurs with probability $\theta/k$ and requires dedicated resource $i$. Its correct response is allocation if $C_i=1$ and infeasibility otherwise. With probability $\delta=1-\theta$, a flexible request accepts the backup, and its correct response is location $J$.

Subsequent source access checks the availability of a named resource. A query either reveals one $C_i$ or asks whether $J=j$ at a specified location $j$. Each query costs one record read. The initial report can therefore preserve either the dedicated-resource availability or the backup location. The two descriptions use exactly the same report capacity. This setting represents an initial inspection followed by resource-specific status requests, as in inventory locations, service providers, or reservation records.

The same comparison has a distributed implementation. One specialist observes $C$, another observes $J$, and the product of their report alphabet sizes is at most $N$. Sending $C$ uses the report budget at the first specialist, and sending $J$ uses it at the second. The result below derives lower bounds for arbitrary encoders of $(C,J)$. These bounds also apply to the separate specialists. There are no restrictions on which resource labels the coordinator may name in its final answer.

For a report $f$, let $R_0(f)$ be its minimum failure probability without queries. Write $R_0^*=\min_fR_0(f)$. An immediately optimal report attains this minimum. After selecting such a report, the coordinator can design its subsequent query strategy optimally. Thus the comparison includes the best possible recovery rule for the report selected by the immediate criterion.

\begin{theorem}
\label{thm:selection}
Let $\lambda=\theta/k>\delta>0$, $N=2^k$, $q\in\mathbb Z_+$, and $\varepsilon,\eta\ge0$.
\begin{enumerate}[label=(\roman*)]
\item The minimum immediate failure probability is $R_0^*=\delta(1-N^{-1})$.
Every minimizing deterministic report is a one-to-one relabeling of $C$.
\item After any immediately optimal report, the minimum failure probability with a cap of $q$ queries and the minimum expected query cost at failure probability at most $\varepsilon$ are, respectively,
\begin{align}
 R_q^{\rm imm}&=\delta\left(1-\frac{q+1}{N}\right)_+,\label{eq:imm-risk}\\
 C^{\rm imm}(\varepsilon)&=\frac{N+2}{2}\left[\delta(1-N^{-1})-\varepsilon\right]_+.
 \nonumber
\end{align}
The expected-cost expression allows randomized stopping. Guaranteed correctness requires $N-1$ queries in the worst case.
\item If $R_0(f)\le R_0^*+\eta$, every subsequent rule with at most $q$ queries satisfies
\begin{equation}
 R(f,\gamma)\ge\delta\left[1-\frac{q+1}{N}-\frac{\eta}{\lambda-\delta}\right]_+.
 \label{eq:near-risk}
\end{equation}
\item Jointly choosing the report and query rule reduces the optimal worst-case count for guaranteed correctness to one. The report $M=J$ attains it, with expected query cost $\theta$ and immediate failure probability $\theta/2$.
\item For $k\ge2$ and $1\le q\le N-2$, there is a deterministic report with $R_0(f)=R_0^*+(\lambda-\delta)/N^2$ and a recovery rule attaining~\eqref{eq:near-risk}. Therefore its coefficient of $\eta$ cannot be uniformly reduced.
\end{enumerate}
The statements hold for private report randomization. With an independent shared seed $U$, the characterization in part~(i) applies to the coordinator's observation $(M,U)$ conditional on $U$.
\end{theorem}

Theorem~\ref{thm:selection} establishes a clear separation between reports that are optimal for immediate accuracy and reports that support low-cost reliable completion. It also quantifies how much immediate accuracy can be sacrificed without eliminating the subsequent recovery burden. Thus, immediate success is not in general a reliable proxy for the information needed to resolve later requests efficiently. The separation in Theorem~\ref{thm:selection} can be understood through the report-capacity constraint. A decoder that preserves the common dedicated requests has used its entire report capacity to distinguish the $N$ availability vectors. The backup location remains uniform after each report. A location query can then eliminate one candidate or find the backup. In contrast, the report $J$ resolves the expensive search at the initial inspection. Each later dedicated request requires only one availability check.

\subsection{Consequences for Report Evaluation}

Note that part~(iii) of Theorem~\ref{thm:selection} bounds the recovery risk when a report sacrifices some immediate accuracy, and part~(v) shows that this bound is attained. Changing the report $C$ on a single source pair attains~\eqref{eq:near-risk} at $\eta=(\lambda-\delta)/N^2$. Thus even a small departure from immediate optimality has a precisely quantifiable effect on the subsequent risk.

The cost difference also persists along a sequence of instances whose minimum immediate failure tends to zero. Choose $\delta=1/[2(k+1)]$ and $\theta=1-\delta$.
Then $\lambda-\delta=1/(2k)$. An immediately optimal report has failure probability below $1/[2(k+1)]$, but complete recovery requires $2^k-1$ queries in the worst case and expected cost $C^{\rm imm}(0)=(2^k-1)(2^k+2)/[4(k+1)2^k]$.
The report $J$ has immediate failure approaching one half and becomes fully reliable after at most one query. Its expected query cost is less than one. Consequently, the ratio between the cost after an immediately optimal report and the minimum expected cost under joint design grows at least on the order of $2^k/k$.

Allowing a small loss of immediate accuracy can still leave recovery expensive. If $R_0(f)\le R_0^*+1/(8k)$, part~(iii) implies that guaranteed correctness requires $q\ge(3/4)2^k-1$.
More generally, attaining failure probability at most $\varepsilon$ requires $q\ge N(1-\varepsilon/\delta-\eta/(\lambda-\delta))-1$.
This expression identifies the accuracy tolerance over which the report-selection effect persists. Choosing a report for later recovery can require accepting lower immediate accuracy while retaining facts whose later acquisition is expensive.

These comparisons use checks of individual resources. Suppose a complete read of $(C,J)$ is also available at cost $c_{\mathrm{full}}$. Then the immediately optimal report can answer dedicated requests directly and read the source on a backup request. Its expected recovery cost is therefore at most $\delta c_{\mathrm{full}}$, and its worst-case cost is at most $c_{\mathrm{full}}$. A specialist retaining the full record can supply $J$ through the same access option. The report $J$ instead incurs expected cost $\theta$ under unit-cost dedicated checks. These attainable costs compare retaining a hard-to-locate identity with retrieving it later through the available interface.

\section{Numerical Experiments}
\label{sec:numerical}

In the numerical experiments, we first compare the theoretical values with independent enumeration on small resource-allocation instances. We then examine reporting roles in larger service systems before implementing reporting, planning, and review with language agents. The evaluation concerns complete decisions, including infeasibility responses, and the source access needed to produce them.

\subsection{Small Private-Resource Allocations}
\label{subsec:verification}

The first experiment asks whether restricting each source to a location or status report preserves the unrestricted joint optimum and whether the proposed access calculation recovers its value. We use small instances for which every permitted reporting design can be evaluated. They contain two to four private sources, each with two to four possible resource identities, and report budgets of two to four symbols. The capacity instances include identical sources and sources with unequal workloads and reading costs. Eight further instances specify feasible combinations of sources in four other ways. For example, a plan may require either sources $\{1,2\}$ or sources $\{3,4\}$, with the actual identities named in either case.

To obtain an independent reference value for each instance, we enumerate all local set partitions within the report budgets. After each report outcome, the optimizer stores the source states consistent with that outcome and minimizes expected cost over subsequent reads. Averaging over report outcomes gives the cost of that design. We compare the unrestricted minimum with the two-form minimum in Theorem~\ref{thm:report-reduction} and, for capacity instances, with the calculation based on source-type counts in Theorem~\ref{thm:capacity-joint}. Each report-design evaluation fixes one joint choice of local partitions.

\begin{table}[!htbp]
\centering
\caption{Enumeration of private-resource reporting designs. Error is the largest absolute difference between the unrestricted optimum and the corresponding theoretical calculation.}
\label{tab:partition-enumeration}
{\footnotesize\renewcommand{\arraystretch}{0.85}
\begin{tabular}{lrrr}
\toprule
Allocation instances & Cases & Report designs & Maximum error\\
\midrule
Identical sources & 40 & 94,392 & $5.6\times10^{-16}$\\
Unequal workloads and costs & 40 & 80,728 & $8.9\times10^{-16}$\\
Other feasible-set families & 8 & 32,512 & $0$\\
\bottomrule
\end{tabular}
}
\end{table}

Table~\ref{tab:partition-enumeration} reports agreement in all 88 instances, comprising 207,632 report-design evaluations. The two-form restriction in Theorem~\ref{thm:report-reduction} therefore recovers the unrestricted optimum across both capacity and non-capacity allocation requirements. For the capacity instances, the count-based calculation in Theorem~\ref{thm:capacity-joint} also matches the full-state enumeration, confirming that it preserves the value of optimal adaptive recovery despite aggregating individual source states.

\subsection{Common Capacity and Identical Private Sources}
\label{subsec:report-experiments}

The next experiment asks when specialists with identical source laws should receive different reporting instructions. Each task can use either one unit of common capacity or its own private resource. We vary the number of tasks from eight to 128, the shortage from one unit to the entire workload, and the report capacity from two to three symbols. Three distributions over absence and six resource identities represent diffuse availability, scarce availability, and availability concentrated on one identity. Availability probability is 0.6 for the diffuse and concentrated laws and 0.2 for the scarce law. Together these choices give 144 settings.

For each setting, we compare all-location reports, all-status reports, and the jointly optimized allocation of the two roles. Every design receives its own optimal access rule and must complete all source states correctly. Expected reads are calculated by summing over report outcomes, so the comparison measures the effect of reporting instructions under the same completion requirement. Figure~\ref{fig:pool-joint} shows the resulting costs and reporting roles.

\begin{figure}[!htbp]
\centering
\includegraphics[width=\linewidth]{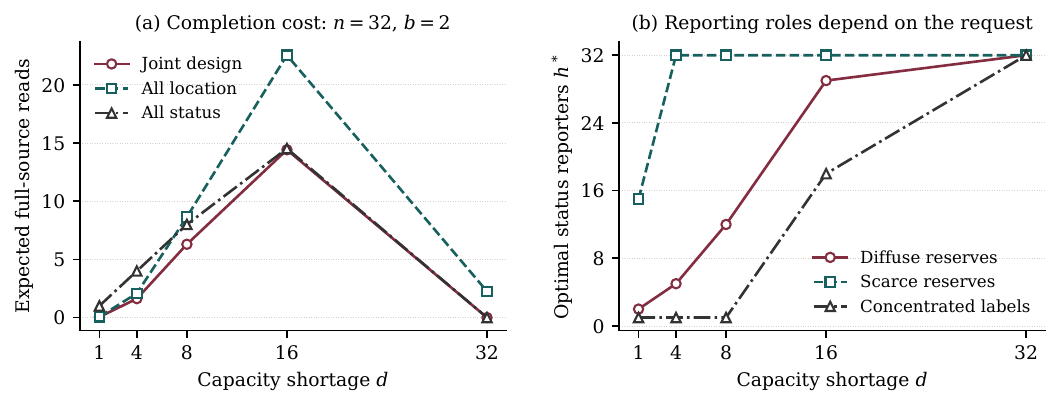}
\caption{Expected full-source reads and optimal reporting roles for identical sources. Panel (a) uses diffuse availability. Panel (b) compares diffuse, scarce, and concentrated source laws.}
\label{fig:pool-joint}
\end{figure}

At intermediate shortages, location reports leave feasibility uncertain, and status reports leave more resource names unread. Joint design balances these two sources of later work. Under the diffuse law with $n=128$, $d=32$, and $b=2$, its expected cost is 24.111 reads, compared with 34.560 for all-location and 32.000 for all-status reports. The saving relative to the better uniform design is 24.65\%. Increasing $b$ to three reduces the optimum to 9.528 and the number of status reporters from 44 to 24. In total, 87 settings use both reporting roles. At high shortages, absence reports often establish infeasibility immediately, explaining the eventual decline in expected cost.

\subsection{Common Capacity with Unequal Workloads and Reading Costs}
\label{subsec:hetero-experiments}

We next ask how source heterogeneity changes both reporting and access. The service systems have two equally sized task types and between eight and 24 tasks in total. Private-resource availability probabilities are 0.6 and 0.3 for the two types. A type-two task can require twice the common capacity of a type-one task, and reading its private-resource record can cost up to ten times as much. We vary the shortage and report budget along with these workload and cost ratios, giving 144 settings.

We compare three designs. Typewise uniform reporting assigns the same role to every specialist of a type and chooses each type's role with minimum-cost access. Joint design additionally permits different roles within a type. The third design retains the joint reports but minimizes the number of reads, breaking ties by expected reading cost. All three use the same source laws and require correct complete assignments or infeasibility responses. We calculate expected reading costs exactly, averaging over all report outcomes.

\begin{figure}[!htbp]
\centering
\includegraphics[width=\linewidth]{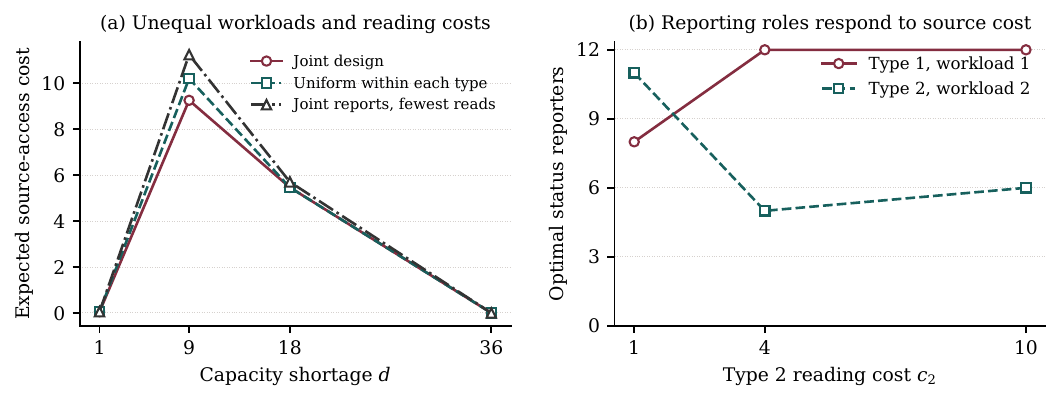}
\caption{Joint design with two source types and unequal workloads. Panel (a) compares reporting and access choices at reading costs $(1,4)$. Panel (b) shows how relative reading cost changes the optimal reporting roles. Both panels use $n=24$ and binary reports.}
\label{fig:hetero-joint}
\end{figure}

Joint design improves on the best typewise uniform rule in 65 settings. Among these settings, the median saving is 8.29\% and the maximum is 24.21\%. For example, let $n=24$, $b=2$, $d=9$, workloads be $(1,2)$, and reading costs be $(1,10)$. Joint reporting costs 14.568, compared with 17.152 for the best typewise uniform design. Its status-report counts are $(12,6)$. Thus all type-one specialists report availability, but type-two specialists split between the two roles. Figure~\ref{fig:hetero-joint}(a) also shows a higher cost under minimum-read access. Reducing the number of reads can therefore increase expenditure when it directs reading toward more expensive sources.

In Appendix~\ref{ec:report-experiments}, we provide an additional comparison that examines minimum-read and minimum-cost access with three-symbol reports.

\subsection{Service Assignments with Language Agents}
\label{subsec:implemented}

These experiments ask how reporting instructions affect completion when language agents extract local facts and choose subsequent reads themselves. Each specialist reads one natural-language service record. A coordinator receives the selected report fields and can request complete records before returning an assignment or infeasibility. An independent checker then verifies that decision.

We use two instance families. The first contains 24 assignments with four to eight identical appointments, two capacity-shortage levels, and unit reading costs. Four states are sampled independently in each of six settings. The second contains 12 assignments with six or eight tasks split equally between two types. These types have different source laws, workloads $(1,2)$, and unequal reading costs. Six states are sampled in each of two settings.

For each family, every specialist supplies both binary report fields. The three paired designs reuse these outputs, transmitting only the selected field. Identical-workload instances compare all-location, all-status, and joint reports, each with minimum-read access. Unequal-workload instances compare joint cost design, the best typewise uniform reports with minimum-cost access, and reports minimizing cost under minimum-read access. Reporting roles are chosen under the source laws before sampling. Every design receives a fresh coordinator without an access algorithm or corrective feedback. The unequal-workload coordinators also receive the source laws, workloads, and reading costs. Each final decision must be correct for every state consistent with the reports and subsequent responses.

\begin{table}[!htbp]
\centering
\caption{Correct final decisions and mean source-access costs on paired service instances. Total cost includes initial inspections. For identical workloads, cost equals the number of reads.}
\label{tab:llm-trials}
{\footnotesize\renewcommand{\arraystretch}{0.85}
\begin{tabular}{lrrr}
\toprule
Report design & Correct decisions & Recovery cost & Total cost\\
\midrule
\multicolumn{4}{l}{Identical workloads and unit reading costs, 24 instances}\\
All-location & 24/24 & 1.458 & 7.458\\
All-status & 24/24 & 1.583 & 7.583\\
Joint & 24/24 & 1.000 & 7.000\\
\midrule
\multicolumn{4}{l}{Unequal workloads and reading costs, 12 instances}\\
Joint cost design & 12/12 & 6.083 & 32.583\\
Uniform within types & 12/12 & 5.833 & 32.333\\
Minimum-read design & 12/12 & 8.250 & 34.750\\
\bottomrule
\end{tabular}
}
\end{table}

All 228 specialists report both fields correctly, and all 108 coordinator decisions satisfy the completion requirement. Table~\ref{tab:llm-trials} reports their costs. On identical workloads, joint reporting saves 0.458 reads relative to all-location reports and 0.583 relative to all-status reports. Every run agrees with its algorithmic reading cost.

On unequal workloads, all 58 observed reading choices attain the assigned objective's optimum conditional on the information then available. Minimum-read design uses 1.750 reads on average, compared with 1.833 for joint cost design, but incurs higher cost. Typewise uniform reporting costs less than joint design on these draws. Under the source laws, the corresponding optimal algorithms have exact expected costs of 4.029 and 3.901, respectively, compared with 4.460 for minimum-read design. Thus the cost objective affects later reads, and optimizing expected cost does not ensure lower cost on every sampled set.

\subsection{Review of Service Assignments}
\label{subsec:review-experiment}

This experiment asks how reviewing a proposed assignment changes its reliability when the reviewer can either reuse its reports or acquire further source information. We use eight prespecified service instances with four to eight tasks, common capacity, private providers, and all-location reports. Each procedure receives the source probabilities and seeks to maximize correct completion.

A direct planner first proposes a complete assignment without source access. Two fresh reviewers receive that same candidate and its reports. One uses those observations alone, while the other can request complete source records. Neither receives the checker's assessment. Table~\ref{tab:review-trials} reports final correctness and reading effort, including the common initial inspections.

\begin{table}[!htbp]
\centering
\caption{Correct decisions and mean reading effort for eight service assignments. Additional reads follow the initial reports.}
\label{tab:review-trials}
{\footnotesize\renewcommand{\arraystretch}{0.85}
\begin{tabular}{lrrr}
\toprule
Procedure & Correct decisions & Additional reads & Total reads\\
\midrule
Direct planning & 5/8 & 0 & 5.750\\
Review using reports & 5/8 & 0 & 5.750\\
Review with source access & 8/8 & 2.250 & 8.000\\
\bottomrule
\end{tabular}
}
\end{table}

To assess whether further calculation could improve these decisions, we enumerate the source states consistent with each report and their conditional probabilities. Every direct and report-only decision attains the minimum conditional error. The mean of these error probabilities is 0.454, and five of the eight sampled states result in correct decisions. Review with source access completes all eight cases correctly. Thus the observed errors persist even when the reported information is used optimally, and further reading provides facts that improve these decisions.

\section{Conclusion}
\label{sec:conclusion}

This paper studies how local reports and subsequent source access affect language agents' ability to complete a resource-constrained plan. Its main structural result reduces arbitrary local report partitions to two forms for a broad class of allocation requirements. For shared capacity and a fixed number of source types, this reduction permits exact joint optimization of reporting and access under unequal workloads and reading costs.

Our analysis points to several important considerations for designing and evaluating LLM-based multi-agent systems.

First, discussions of whether more agents improve planning should consider what their communication contributes to the complete decision. Additional specialists can process more records, but equally accurate reports of equal length can leave very different recovery costs. In the synchronized example, the useful information concerns compatibility between resources, which individual report accuracy and length do not capture. In the allocation model, even identical specialists can optimally receive different reporting instructions. These results explain how specialization can improve planning through the different facts that specialists retain and how those facts together support a feasible plan.

Second, ongoing discussions of reflection and verification in LLM-based multi-agent planning concern both the use of available information and the acquisition of missing information. When agents draft and review a plan, a reviewer can correct a calculation or identify an assignment overlooked in the existing reports. Returning to a source can additionally reveal an omitted resource. To assess the room for improvement from calculation alone, our analysis compares a decision with the best one supported by the reviewer's information. The language-agent comparison illustrates the role of source access when initial decisions already attain that benchmark. Evaluation should therefore account for which information a review stage receives and how effectively it uses it.

Third, communication that favors immediate completion can make the remaining cases costly to finish. This is relevant when systems are evaluated mainly on whether they answer without another tool call. In our allocation family, the resulting recovery-cost ratio grows without bound, even with complete-record access and many feasible alternatives. The joint objective instead values a report through both stages. It explains why some specialists should identify usable resources immediately while others should indicate where subsequent reading will be useful. Thus an apparently strong immediate response rate can conceal an expensive continuation.

Last, the complete planning episode provides a common basis for comparing agent architectures. Reporting establishes what is known, calculation uses that information, and source access supplies what remains missing. Each stage should be evaluated against the same complete decision and explicit access costs. Minimizing the number of reads can be costly when sources differ in retrieval cost. The joint allocation model accounts for these costs when selecting reporting roles and subsequent reads. Comparisons should include initial reporting effort as well as recovery. Communication and source access should be designed jointly and judged by the reliable completion they achieve for the total effort required.

\bibliographystyle{plainnat}
\bibliography{references}

\begin{thebibliography}{19}
\providecommand{\natexlab}[1]{#1}
\providecommand{\url}[1]{\texttt{#1}}
\expandafter\ifx\csname urlstyle\endcsname\relax
  \providecommand{\doi}[1]{doi: #1}\else
  \providecommand{\doi}{doi: \begingroup \urlstyle{rm}\Url}\fi

\bibitem[Ahmadian et~al.(2026)Ahmadian, Cohen, and
  Roghani]{ahmadian2026sparsification}
Sara Ahmadian, Edith Cohen, and Mohammad Roghani.
\newblock Stochastic matching via local sparsification.
\newblock In \emph{Advances in Neural Information Processing Systems}, 2026.
\newblock To appear.

\bibitem[Ben-Dov(1981)]{BenDov1981}
Yosi Ben-Dov.
\newblock Optimal testing procedures for special structures of coherent
  systems.
\newblock \emph{Management Science}, 27\penalty0 (12):\penalty0 1410--1420,
  1981.

\bibitem[Bertsimas and Kallus(2020)]{bertsimas2020prescriptive}
Dimitris Bertsimas and Nathan Kallus.
\newblock From predictive to prescriptive analytics.
\newblock \emph{Management Science}, 66\penalty0 (3):\penalty0 1025--1044,
  2020.

\bibitem[Bui et~al.(2026)Bui, Li, and Liu]{bui2026himap}
The~Viet Bui, Wenjun Li, and Yong Liu.
\newblock {HiMAP-Travel}: Hierarchical multi-agent planning for long-horizon
  constrained travel.
\newblock In \emph{Proceedings of the 43rd International Conference on Machine
  Learning}, volume 306, pages 10219--10252. PMLR, 2026.

\bibitem[Cemri et~al.(2025)Cemri, Pan, Yang, Agrawal, Chopra, Tiwari, Keutzer,
  Parameswaran, Klein, Ramchandran, Zaharia, Gonzalez, and
  Stoica]{cemri2025multiagent}
Mert Cemri, Melissa~Z. Pan, Shuyi Yang, Lakshya~A. Agrawal, Bhavya Chopra,
  Rishabh Tiwari, Kurt Keutzer, Aditya Parameswaran, Dan Klein, Kannan
  Ramchandran, Matei Zaharia, Joseph~E. Gonzalez, and Ion Stoica.
\newblock Why do multi-agent {LLM} systems fail?
\newblock In \emph{Advances in Neural Information Processing Systems},
  volume~38, 2025.

\bibitem[Chen et~al.(2026)Chen, Zhu, Zhang, and Li]{chen2026coact}
Haorui Chen, Yuancheng Zhu, Yitong Zhang, and Jia Li.
\newblock {CoACT}: Action-preserving observation compression for coding agents.
\newblock \emph{arXiv preprint arXiv:2607.02911}, 2026.
\newblock URL \url{https://arxiv.org/abs/2607.02911}.

\bibitem[Doshi et~al.(2010)Doshi, Shah, M{\'e}dard, and
  Effros]{doshi2010functional}
Vishal Doshi, Devavrat Shah, Muriel M{\'e}dard, and Michelle Effros.
\newblock Functional compression through graph coloring.
\newblock \emph{IEEE Transactions on Information Theory}, 56\penalty0
  (8):\penalty0 3901--3917, 2010.

\bibitem[Elmachtoub and Grigas(2022)]{elmachtoub2022smart}
Adam~N. Elmachtoub and Paul Grigas.
\newblock Smart {``Predict, then Optimize''}.
\newblock \emph{Management Science}, 68\penalty0 (1):\penalty0 9--26, 2022.

\bibitem[Gou et~al.(2024)Gou, Shao, Gong, Shen, Yang, Duan, and
  Chen]{gou2024critic}
Zhibin Gou, Zhihong Shao, Yeyun Gong, Yelong Shen, Yujiu Yang, Nan Duan, and
  Weizhu Chen.
\newblock {CRITIC}: Large language models can self-correct with
  tool-interactive critiquing.
\newblock In \emph{International Conference on Learning Representations}, 2024.

\bibitem[Hao et~al.(2025)Hao, Chen, Zhang, and Fan]{hao2025formal}
Yilun Hao, Yongchao Chen, Yang Zhang, and Chuchu Fan.
\newblock Large language models can solve real-world planning rigorously with
  formal verification tools.
\newblock In \emph{Proceedings of the 2025 Conference of the Nations of the
  Americas Chapter of the Association for Computational Linguistics}, pages
  3434--3483, 2025.

\bibitem[Hong et~al.(2024)Hong, Zhuge, Chen, Zheng, Cheng, Wang, Zhang, Wang,
  Yau, Lin, Zhou, Ran, Xiao, Wu, and Schmidhuber]{hong2024metagpt}
Sirui Hong, Mingchen Zhuge, Jonathan Chen, Xiawu Zheng, Yuheng Cheng, Jinlin
  Wang, Ceyao Zhang, Zili Wang, Steven Ka~Shing Yau, Zijuan Lin, Liyang Zhou,
  Chenyu Ran, Lingfeng Xiao, Chenglin Wu, and J{\"u}rgen Schmidhuber.
\newblock {MetaGPT}: Meta programming for a multi-agent collaborative
  framework.
\newblock In \emph{International Conference on Learning Representations}, 2024.

\bibitem[Huang et~al.(2024)Huang, Chen, Mishra, Zheng, Yu, Song, and
  Zhou]{huang2024selfcorrect}
Jie Huang, Xinyun Chen, Swaroop Mishra, Huaixiu~Steven Zheng, Adams~Wei Yu,
  Xinying Song, and Denny Zhou.
\newblock Large language models cannot self-correct reasoning yet without
  external feedback.
\newblock In \emph{International Conference on Learning Representations}, 2024.

\bibitem[Javdani et~al.(2014)Javdani, Chen, Karbasi, Krause, Bagnell, and
  Srinivasa]{javdani2014decision}
Shervin Javdani, Yuxin Chen, Amin Karbasi, Andreas Krause, Drew Bagnell, and
  Siddhartha Srinivasa.
\newblock Near optimal bayesian active learning for decision making.
\newblock In \emph{Proceedings of the Seventeenth International Conference on
  Artificial Intelligence and Statistics}, volume~33 of \emph{Proceedings of
  Machine Learning Research}, pages 430--438, 2014.

\bibitem[Kim et~al.(2026)Kim, Gu, Park, Park, Schmidgall, Heydari, Yan, Zhang,
  Zhuang, Liu, Malhotra, Liang, Park, Yang, Xu, Du, Patel, Althoff, McDuff, and
  Liu]{kim2025scaling}
Yubin Kim, Ken Gu, Chanwoo Park, Chunjong Park, Samuel Schmidgall, A.~Ali
  Heydari, Yao Yan, Zhihan Zhang, Yuchen Zhuang, Yun Liu, Mark Malhotra,
  Paul~Pu Liang, Hae~Won Park, Yuzhe Yang, Xuhai Xu, Yilun Du, Shwetak Patel,
  Tim Althoff, Daniel McDuff, and Xin Liu.
\newblock Capable language models can outgrow the benefits of collaboration.
\newblock \emph{Nature Machine Intelligence}, 8:\penalty0 1157--1172, 2026.

\bibitem[Mohanty et~al.(2026)Mohanty, Bolar, Patil, Ganesh, Chamberland, and
  Parag]{mohanty2026matching}
Moonmoon Mohanty, Gautham Bolar, Preetam Patil, Ayalvadi Ganesh, Jean-Francois
  Chamberland, and Parimal Parag.
\newblock Bipartite matching under communication constraints.
\newblock \emph{arXiv preprint arXiv:2604.10744}, 2026.
\newblock URL \url{https://arxiv.org/abs/2604.10744}.

\bibitem[Permuter and Weissman(2011)]{permuter2011vending}
Haim~H. Permuter and Tsachy Weissman.
\newblock Source coding with a side information vending machine.
\newblock \emph{IEEE Transactions on Information Theory}, 57\penalty0
  (7):\penalty0 4530--4544, 2011.

\bibitem[Radner(1962)]{radner1962team}
Roy Radner.
\newblock Team decision problems.
\newblock \emph{Annals of Mathematical Statistics}, 33\penalty0 (3):\penalty0
  857--881, 1962.

\bibitem[Wu et~al.(2024)Wu, Bansal, Zhang, Wu, Li, Zhu, Jiang, Zhang, Zhang,
  Liu, Awadallah, White, Burger, and Wang]{wu2023autogen}
Qingyun Wu, Gagan Bansal, Jieyu Zhang, Yiran Wu, Beibin Li, Erkang Zhu,
  Li~Jiang, Xiaoyun Zhang, Shaokun Zhang, Jiale Liu, Ahmed~Hassan Awadallah,
  Ryen~W. White, Doug Burger, and Chi Wang.
\newblock {AutoGen}: Enabling next-gen {LLM} applications via multi-agent
  conversations.
\newblock In \emph{First Conference on Language Modeling}, 2024.

\bibitem[Xie et~al.(2024)Xie, Zhang, Chen, Zhu, Lou, Tian, Xiao, and
  Su]{xie2024travelplanner}
Jian Xie, Kai Zhang, Jiangjie Chen, Tinghui Zhu, Renze Lou, Yuandong Tian,
  Yanghua Xiao, and Yu~Su.
\newblock {TravelPlanner}: A benchmark for real-world planning with language
  agents.
\newblock In \emph{Proceedings of the 41st International Conference on Machine
  Learning}, volume 235, pages 54590--54613. PMLR, 2024.

\end{thebibliography}
\clearpage
\appendix
\numberwithin{equation}{section}
\numberwithin{table}{section}
\numberwithin{figure}{section}
These appendices follow the main text. Sections~\ref{ec:synchronized} and~\ref{ec:capacity-joint} prove the results in Sections~\ref{sec:query} and~\ref{sec:capacity-assignments}, respectively. Section~\ref{ec:selection} proves Theorem~\ref{thm:selection}. Section~\ref{app:experiments} provides additional settings for the numerical experiments in Section~\ref{sec:numerical}.

\section{Proofs of the Results in Section~\ref{sec:query}}
\label{ec:synchronized}

\begin{proof}[Proof of Lemma~\ref{thm:cell}]
The two reports leave the source states in $A$ and $B$. Write $H=A\cap B$ and $h=|H|$ for their common candidate labels. An empty intersection establishes infeasibility, and identical singleton cells establish the assignment. If $A\subseteq B$, check the first source on all but one label of $A$, stopping at a positive answer or inferring the last label. Checking this label at the second source uses at most $h$ checks in total. The case $B\subseteq A$ is symmetric. Otherwise, check all $h$ common labels at one source. Follow a positive answer by checking that label at the other source, and return $\bot$ after all negative answers. This uses at most $h+1$ checks.

For the lower bounds, consider the states consistent with the reports and all query answers. Since each query concerns one source, this set remains a Cartesian product. A common correct action exists exactly when the remaining cells have empty intersection or are the same singleton. Every other cell requires a query. For $h\ge2$, follow negative answers until one common label remains. Stopping with two common labels $i,j$ is impossible because $(i,i)$ and $(i,j)$ require different actions. The query removing the $(h-1)$st common label occurs at depth at least $h-1$. On its positive branch, the other coordinate can still equal that label or the remaining common label, requiring another query. Hence the depth is at least $h$.

If both exclusive sets are nonempty, any query removing no common label before the $(h-1)$st common removal increases this bound to $h+1$. Otherwise, the first $h-1$ queries remove only common labels. With remaining label $i$, choose $x_0\in A\setminus B$ and $y_0\in B\setminus A$. The states $\{i,x_0\}\times\{i,y_0\}$ remain possible. No single query distinguishes the sole feasible state $(i,i)$ from all others, so two additional queries are necessary. This argument also applies initially when $h=1$.

To include randomized rules, let $U$ be a random tape independent of the sources and fix 
$U=u$. The union, over finitely many states, of tapes violating correctness or the cap is null. Every remaining tape gives a deterministic rule with the same guarantees. Averaging over the random choices gives the distributional bounds. \Halmos
\end{proof}

\begin{proof}[Proof of Theorem~\ref{thm:capacity}]
Represent the two report partitions by a bipartite graph, with one vertex per cell and an edge for each nonempty intersection. Its weights $h_{ij}=|A_i\cap B_j|$ sum to $n$. Assigning one distinct label per unit weight and reporting its endpoints realizes every such graph. Moreover, $A_i\subseteq B_j$ exactly when vertex $i$ has degree one adjacent to $j$. For $q\ge1$, Lemma~\ref{thm:cell} bounds an edge weight by $q$ if either endpoint has degree one, and by $q-1$ otherwise.

At $q=0$, edges must be disjoint with weight one, giving $\min\{a,b\}$. At $q=1$, every edge has weight one and touches a degree-one vertex, so nonempty components are stars. If one part has capacity one, the maximum is $\max\{a,b\}$. For $a,b\ge2$, two or more components give at most $a+b-2$ edges, and one star gives at most $\max\{a,b\}\le a+b-2$. Two disjoint stars with opposite-side centers and $a-1,b-1$ leaves attain the bound.

For $q\ge2$, put $w=q-1$. If there are $e$ edges and $\ell$ touch a degree-one vertex, total weight is at most $we+\ell$. Let $x,y$ count vertices of degree at least two in the respective parts. Since $e-\ell\le xy$ and $\ell\le a+b-x-y$, we have $we+\ell\le wxy+q(a+b-x-y)$ for $x,y\ge1$. This expression is affine in either variable with the other fixed. Its maximum over $[1,a]\times[1,b]$ is the larger of the corner values $wab$ and $q(a+b-1)-1$. If $x=0$ or $y=0$, weight is at most $q\max(a,b)$, no larger than the second value for $a,b\ge2$.

The complete bipartite graph with weights $q-1$ attains the first value. Joining two opposite-side centers with weight $q-1$ and every other vertex to the opposite center with weight $q$ attains the second. A single-vertex part uses a star with weights $q$. Removing surplus labels and empty cells cannot increase cost, proving the formula for $Q^*$. Generating occupied edges and storing label-to-cell maps and intersection lists takes $O(n+a+b)$ space and operations.\Halmos
\end{proof}

\begin{proof}[Proof of Proposition~\ref{prop:balanced}]
We first compare the worst-case checks under the two report designs. Identical cells have size $rs$ and depth $rs$. Row and column cells have intersections of size $s$ and nonempty exclusive sets, giving depth $s+1$. At $q=s+1$, capacity is at least $sr^2=n$. At $q=s\ge2$, both $(s-1)r^2$ and $s(2r-1)-1$ are below $n$. At $s=1$, one-query capacity is $2r-2<r^2$. Thus $s+1$ is optimal. Both reports have $r$ equiprobable messages and uniform conditional sources of size $rs$, giving the stated local information measures.

Conditional on equality within an identical block, the common label is uniform on $h=rs$ labels. Revealing equality can only reduce cost. Each query then tests one candidate label, so guaranteed identification has mean cost at least $[1+2+\cdots+(h-1)+(h-1)]/h=(h-1)(h+2)/(2h)$. With cap $q$, the path on which every query returns a negative answer tests at most $q$ labels, and its terminal answer identifies at most one more. Success is therefore at most $\min\{1,(q+1)/h\}$. Multiplying by the equality probability $\pi$ gives both distributional bounds. Crossed reports complete every state within $s+1$ checks by Lemma~\ref{thm:cell}.\Halmos
\end{proof}

\section{Proofs of the Results in Section~\ref{sec:capacity-assignments}}
\label{ec:capacity-joint}

\begin{proof}[Proof of Theorem~\ref{thm:report-reduction}]
Fix deterministic reports. We first characterize when reading can stop. After any history, let $N$ contain the sources with known available names and $P$ those that can still be positive. Correct allocation is possible exactly when $N\in\mathcal F$. An unread nonsingleton cell contains a state without its proposed resource name. Independence and full support make that state compatible with the other observations, so only named resources can be assigned. Upward closure gives the condition on $N$. Correct infeasibility is possible exactly when $P\notin\mathcal F$. If $P\in\mathcal F$, the state with all possibly positive sources available is admissible and feasible. Otherwise every realized available set, being a subset of $P$, is infeasible.

To compare report cells, recall that resource contributions depend only on source indices. Every nonsingleton positive cell establishes availability and costs $c_i$ to resolve, and its eventual identity reveals nothing about other sources. Thus reads can be chosen using availability and source indices, retaining identities for the final allocation. Conditional on other reports, all nonsingleton positive-only cells at a source have the same continuation value.

We now change one source report at a time. Fix source $i$ and the other reporting rules. Merge any nonsingleton positive-only cells, preserving their continuation value. Let $D_1,D_0,D_+$ denote optimal costs conditional on source $i$ being named, absent, or available but unnamed. These values do not depend on the positive identity, and $D_1\le D_+$ because a named resource can simulate any unnamed-resource continuation while omitting its own read.

Suppose a nonsingleton positive-only cell exists. Move all positive states from the zero-containing cell into that cell, leaving singleton zero. If $r$ is the positive mass of the old mixed cell and $D_\pi^0,D_\pi^+$ are the conditional costs of any continuation $\pi$ there, then $p_{i0}D_0+rD_+\le p_{i0}D_\pi^0+rD_\pi^+$. Knowing absence or availability permits choosing the two continuations separately. Existing positive-only mass retains value $D_+$, and the number of cells does not increase. Splitting positive singletons until all $b_i$ symbols are used improves information and leaves singleton zero, $b_i-2$ positive singletons, and a positive group of at least two states. If no nonsingleton positive-only cell existed, all cells other than the zero-containing cell were positive singletons. Splitting further singletons produces $b_i-1$ positive singletons and a mixed cell containing zero and at least one positive state.

It remains to select the singleton states. For a status report with total positive singleton probability $s$, conditional cost is $sD_1+p_{i0}D_0+(1-p_{i0}-s)D_+$, nonincreasing in $s$. Hence the $b_i-2$ largest positive probabilities are optimal. For a location report, continuations after the mixed cell can be represented by the same collection $\Pi$ of reading rules, recording source $i$ as absent or positive when read. Resource names only fill the final allocation. Expected cost is
\[
sD_1+\min_{\pi\in\Pi}\{p_{i0}D_\pi^0+(1-p_{i0}-s)D_\pi^+\}.
\]
A named source permits simulation of each positive-branch rule, omitting the reading cost of that source if it is queried. Thus $D_1\le D_\pi^+$. Increasing $s$ by $\Delta\ge0$ and using the previous minimizing rule changes cost by at most $\Delta(D_1-D_\pi^+)\le0$. Selecting the $b_i-1$ largest probabilities is therefore optimal. These comparisons hold for every collection of other reports. Replacing each source rule in turn proves that one of the two stated forms is optimal at every source.

Finally, repeated reads can be removed. Fixing all random choices gives, outside a finite union of null sets, a deterministic system correct on every state. Averaging its cost, which is at least the deterministic optimum, proves the randomization claim.
\Halmos
\end{proof}

\begin{proof}[Proof of Theorem~\ref{thm:capacity-joint}]
Apply Theorem~\ref{thm:report-reduction} to $\mathcal F=\{S:\sum_{i\in S}w_i\ge d\}$. Conditional independence is preserved by local reports and by adaptive full-source reads. For each unread grouped location source of type $t$, the conditional positive probability is $q_t$. Every grouped status source is positive. Within a type, resource identities affect only the final output, and sources in the same reported category have identical continuation laws and costs. Thus $(a,u,r)$ determines the continuation value.

The named workload determines whether an allocation can be made, and the possible workload determines whether infeasibility can be declared. These give the two stopping conditions. In every other state, a correct rule must read an unnamed source. Reading an uncertain type-$t$ source costs $c_t$ and either adds one named resource, with probability $q_t$, or removes one possible resource. Reading an unnamed available source costs $c_t$ and adds one named resource with probability one. These are precisely the alternatives in~\eqref{eq:type-value}. Every transition decreases $\sum_t(a_t+u_t)$ by one. Induction on this sum proves the recurrence and that choosing an attaining action is optimal. Retaining the names returned by these reads gives the allocation in Algorithm~\ref{alg:capacity-joint}.

We next average over the initial reports. Sources within a type have the same law and cost, so only the number $h_t$ of status reporters affects the expected cost. Among location reporters, the grouped count is binomial with parameters $(n_t-h_t,\rho_t)$. Among status reporters, the counts of names, unnamed resources, and absences are multinomial with probabilities $(\beta_t,\sigma_t,p_{t0})$. Their independent product gives $\prod_tG_t$. Initially named resources number $n_t-h_t-a_t+j_t$, which proves \eqref{eq:capacity-joint-cost}. The report reduction makes minimization over $h$ unrestricted.

For the computational bound, the counts $a_t+u_t+r_t\le n_t$ give $\binom{n_t+3}{3}$ possibilities per type. Thus the table has $\prod_t\binom{n_t+3}{3}$ states and at most $2K$ actions per state. Summed over $h_t$, outcome triples number $\binom{n_t+4}{4}$, since the five nonnegative counts $a_t$, $n_t-h_t-a_t$, $j_t$, $u_t$, and $h_t-j_t-u_t$ sum to $n_t$ and determine $h_t$ and its outcome. Consequently, evaluating every $C(h)$ takes $O(K\prod_t(n_t+1)^4)$ operations. The required powers, factorials, and multinomial coefficients can be precomputed within the same $O\!\left(K\prod_{t=1}^K (n_t+1)^4\right)$ time bound. For any fixed $h$, evaluating $C(h)$ requires storing only the continuation table $V$ and the sorted probability vectors. Hence the peak space requirement is $O\!\left(\sum_{t=1}^K m_t + \prod_{t=1}^K (n_t+1)^3\right)$. The stated operation count includes the comparisons involving the workloads $w_t$ and the shortage $d$. In particular, the dynamic program does not require a table indexed by $d$ or by the numerical magnitudes of the workloads.
\Halmos
\end{proof}

\begin{proof}[Proof of Corollary~\ref{cor:homogeneous-access}]
If $1\le v\le u$, at least $v$ new names are necessary, and $v$ certain reads attain that number. If $u<v\le a+u$, put $\ell=v-u$. Reading uncertain sources until $\ell$ positives or $a-\ell+1$ negatives uses $T_{a,\ell}$ reads. Its feasible branch has probability $B_{a,\ell}$ and requires $u$ further names. In the Bellman equation, reading an uncertain source first has expected cost
$T_{a,\ell}+uB_{a,\ell}$, because the term $T_{a,\ell}$ is the expected number of
uncertain-source reads needed to determine whether at least $\ell$ of the $a$
uncertain sources are available, and the $u$ unnamed available sources must be
read only on the feasible branch, which has probability $B_{a,\ell}$. In
contrast, reading one unnamed available source first has expected cost
$1+T_{a,\ell}+(u-1)B_{a,\ell}$. Their difference is $1-B_{a,\ell}\ge0$. Induction on $a+u$, with the stated stopping values, proves~\eqref{eq:capacity-value}. The $B,T$ tables use $O(n^2)$ space, and~\eqref{eq:capacity-joint-cost} has $O(n^4)$ terms when summed over $h$.
\Halmos
\end{proof}

\begin{proof}[Proof of Corollary~\ref{cor:pool-one}]
With a one-unit shortage, one named private resource completes the allocation. No name but some positive status report requires one read, an event of probability $(\omega^h-p_0^h)t^{n-h}$. All status reports absent and all location reports grouped has probability $p_0^ht^{n-h}$. The remaining sources have conditional absence probability $p_0/t$, so search until the first positive or complete absence has mean $\sum_{j=0}^{n-h-1}(p_0/t)^j$, with an empty sum equal to zero. These events exhaust the cases requiring reads and give~\eqref{eq:pool-one-cost}.

Since $\omega-t=\xi$, subtraction gives $F_{h+1}-F_h=t^{n-h-1}(\xi\omega^h-p_0^{h+1})$. Its sign is the sign of $(\xi/p_0)(\omega/p_0)^h-1$. The latter increases strictly because $\omega>p_0$. Its first nonnegative value, truncated to $0,\ldots,n$, gives~\eqref{eq:pool-one-optimum}. Equality produces the single stated adjacent tie. If this first nonnegative difference occurs before $h=0$, then
$F_{h+1}-F_h>0$ for all $h=0,\ldots,n-1$, so $F_h$ is strictly increasing and
$h=0$ is the unique optimum. If it occurs after $h=n-1$, then
$F_{h+1}-F_h<0$ for all $h=0,\ldots,n-1$, so $F_h$ is strictly decreasing and
$h=n$ is the unique optimum.
\Halmos
\end{proof}

\begin{proof}[Proof of Corollary~\ref{cor:pool-separation}]
Set $n=k^2$, $a=k^{-3}$, $z=1-k^{-1}$, and $t=1-a$. A binary report names an available resource only through a positive singleton. Since there are at least four positive states, at most one cell can be such a singleton. If there is none, no name is immediately available. A source is known to be absent only when the report assigns state zero to a singleton cell.

We first compare the probability of completion before any read. A zero-read allocation exists exactly when at least one positive singleton is reported. Zero-read infeasibility requires singleton zero at every source. If $s\ge1$ sources have positive singleton reports, correct immediate completion has probability $1-t^s$. If $s=0$, its probability is at most $z^n$, attained by status reports. Bernoulli's inequality gives $1-t^n\ge na/(1+na)=1/(k+1)$, while $z^n\le e^{-k}<1/(k+1)$. Thus every maximizing report design has $s=n$.

It remains to calculate recovery cost under the maximizing design. Conditional on no singleton being reported, sources remain independent with absence probability $z/t$. Correctness requires reading until one positive source is found or all sources are absent. Every order has the same law and is optimal. Its unconditional expected cost is $t^n\sum_{j=0}^{n-1}(z/t)^j=t(t^n-z^n)/(t-z)$, proving~\eqref{eq:pool-separation}. All $n$ reads are necessary on the all-absent state. Since $t^n\ge1-1/k\ge1/2$, $z^n\le e^{-2}$, $t\ge7/8$, and $t-z=k^{-1}(1-k^{-2})$,
\[
\tfrac78(\tfrac12-e^{-2})k
\le\frac{t(t^n-z^n)}{t-z}\le\tfrac43k.
\]
With status reports, a positive report locates a source that completes the allocation after one read. If every report is zero, infeasibility is already determined. Expected cost is therefore $1-z^n<1$, which proves the ratio claim.
\Halmos
\end{proof}

\section{Proof of Theorem~\ref{thm:selection}}
\label{ec:selection}

We prove the five parts in order. The immediate-risk argument identifies the information retained by an optimal report, which then determines the subsequent search problem.

\begin{proof}[Proof of Theorem~\ref{thm:selection}]
\textit{Part (i).} For any report channel, choose a deterministic decoder minimizing expected loss conditional on each message. Let $c(m)$ give its dedicated-resource answers and $j(m)$ its backup answer. Its loss at state $(c,j)$ is $\lambda d_H(c,c(m))+\delta\ind\{j\ne j(m)\}$, where $d_H$ counts differing coordinates. Set $D=\Prob\{c(M)\ne C\}$. Each message correctly specifies both answers for at most one of the $N^2$ equally likely pairs. Thus $\Prob\{c(M)=C,\ j(M)=J\}\le N^{-1}$. Dedicated loss is at least $\lambda D$, and backup success is at most $N^{-1}+D$, giving $R_0(f)\ge\delta(1-N^{-1})+(\lambda-\delta)D$. The report $C$ attains the first term. Since $\lambda>\delta$, optimality requires $D=0$. The $N$ possible vectors and $N$ messages then permit exactly one message per vector, independent of $J$. This also proves the private-randomization claim.

\textit{Part (ii).}
An immediately optimal report reveals $C$ and leaves $J$ uniform. Dedicated requests require no queries. Fix $C$ and the recovery seed, remove repeated checks, and stop at any positive answer. Until the backup is found, the remaining checks follow the path of negative answers. A cap $q$ identifies at most $q$ locations through positive answers and one through its terminal guess. Conditional backup failure is therefore at least $[1-(q+1)/N]_+$. Distinct checks followed by an untested guess attain this bound, giving~\eqref{eq:imm-risk}. Zero error requires $N-1$ checks in the worst case.

To obtain the expected cost bound, consider a path testing $t\in\{0,\ldots,N-1\}$ locations. Its optimal success is $(t+1)/N$, and its expected cost is
\[
 c_t=\sum_{u=1}^t\frac{N-u+1}{N}=\frac{t(2N-t+1)}{2N}
 \ge\frac{N+2}{2}\frac tN.
\]
Equality holds at $t=0,N-1$. Its gain in success over immediate guessing is $t/N$. Conditioning on $C$ and random choices expresses any rule as a mixture of these paths, preserving the inequality. Dedicated errors can be corrected at zero cost. An overall target $\varepsilon<R_0^*$ therefore requires unconditional backup-success gain at least $R_0^*-\varepsilon$, and expected cost at least $(N+2)(R_0^*-\varepsilon)/2$.

On backup requests, mix complete search with immediate guessing, using complete search with probability $\alpha=1-\varepsilon/[\delta(1-N^{-1})]$. This has failure $\varepsilon$ and cost $\delta\alpha(N-1)(N+2)/(2N)$, attaining the expected cost expression in part~(ii). For $\varepsilon\ge R_0^*$, immediate guessing attains zero cost.

\textit{Part (iii).}
We now relate a report's excess immediate risk to its recovery risk. For its deterministic optimal no-query decoder, let $D=\Prob\{c(M)\ne C\}$. Part~(i) gives $R_0(f)\ge R_0^*+(\lambda-\delta)D$, hence $D\le\eta/(\lambda-\delta)$. Fix a recovery seed and message $m$, restricting attention to $C=c(m)$. Dedicated checks now have fixed answers. Only location checks vary, so a rule with cap $q$ identifies at most $q+1$ backup locations in this restricted set. With at most $N$ messages and $N^2$ equally likely pairs, $\Prob\{\text{backup answer correct},\ c(M)=C\}\le(q+1)/N$. This includes stochastic reports, since each counted source--message pair contributes at most $1/N^2$. Success on the remaining event is at most $D$. Overall failure is therefore at least $\delta[1-(q+1)/N-D]_+$, which proves \eqref{eq:near-risk}.

For an independent shared seed $U$, both counting inequalities hold conditionally. Averaging gives the same bounds in terms of $\E D(U)$, so unconditional near-optimality suffices. At exact optimality, conditional risks are all at least $R_0^*$ and average to it. Thus almost every conditional report reveals $C$ by part~(i). Averaging the conditional bounds also proves part~(ii) for shared randomness.

\textit{Part (iv).}
The report $M=J$ uses $N$ symbols. Backup requests can be answered directly,
while each dedicated request requires only one bit check. Its expected cost is
$\theta$. Its immediate failure is $\theta/2$ because dedicated bits remain fair.

Zero-query correctness for every request requires each report to determine both $C$ and $J$, hence distinguish $N^2$ source pairs with only $N$ symbols, which is impossible. For stochastic reports, different pairs must have disjoint message supports. Shared randomness is handled by the finite-tape argument in Section~\ref{ec:synchronized}.

Either the $C$ source reports its $k$ bits and the $J$ source is silent, or conversely, giving distributed implementations with product alphabet size $N$. The lower bounds apply because they permit arbitrary encoders. The $C$ coordinates can also be split among $k$ one-bit specialists.

\textit{Part (v).}
To show that the coefficient in part~(iii) is attained, change the report $C$ on one source pair. Fix $k\ge2$ and $1\le q\le N-2$. Choose vectors $c^0,c^1$ at Hamming distance one and $j^0\in[N]$. Set $f(c,j)=c$, except that $f(c^0,j^0)=c^1$. There are still $N$ messages. Cell $c^0$ has $N-1$ distinct locations. Cell $c^1$ has its original $N$ pairs plus $(c^0,j^0)$. All other cells are unchanged.

The dedicated Bayes decoder remains $c(m)=m$. In the only varying coordinate of cell $c^1$, the majority is $N$ against one. Its dedicated failure contribution is $\lambda/N^2$. Location $j^0$ occurs twice in that cell. Every other cell has distinct locations. The immediate backup decoder is therefore correct on $N+1$ source pairs, giving
\[
 R_0(f)=\lambda/N^2+\delta[1-(N+1)/N^2]
       =R_0^*+(\lambda-\delta)/N^2.
\]
Thus $\eta=(\lambda-\delta)/N^2>0$.

Answer dedicated requests by one bit check. For a backup request in cell $c^1$, test $q$ distinct locations including $j^0$, stopping at a positive answer, and otherwise guess an untested location. This identifies $q+2$ pairs. In every other cell, $q$ distinct checks and an untested guess identify $q+1$ pairs. Cell $c^0$ has enough locations because $q+1\le N-1$. Hence backup success covers $(q+2)+(N-1)(q+1)=N(q+1)+1$ pairs. Dedicated requests never fail, so total failure is $\delta[1-(q+1)/N-1/N^2]$, exactly~\eqref{eq:near-risk}. Any uniformly smaller coefficient of $\eta$ would exceed this attainable risk.\Halmos
\end{proof}

\section{Settings for the Numerical Experiments in Section~\ref{sec:numerical}}
\label{app:experiments}

In this section, we provide additional settings for the numerical experiments in Section~\ref{sec:numerical} and further compare source-access policies that minimize expected reading cost with those that minimize the expected number of reads.

\subsection{Private-Resource Experiments in Section~\ref{subsec:verification}}
\label{ec:verification-experiments}

The private-resource verifier enumerates local partitions without counting message relabelings separately, and then evaluates every choice of unread source. Its heterogeneous instances have two to four sources, report budgets from two to four, workloads up to three, and all integer shortages. Four general feasible-set families have two parameter realizations each, including completion by either of two disjoint pairs and by one source from each of two groups.

\subsection{Shared-Capacity Experiments in Sections~\ref{subsec:report-experiments} and~\ref{subsec:hetero-experiments}}
\label{ec:report-experiments}

The identical-source grid uses $n\in\{8,16,32,64,128\}$, $b\in\{2,3\}$, and distinct shortages in $\{1,\lceil n/8\rceil,\lceil n/4\rceil,n/2,n\}$. Diffuse and scarce laws have $p_0=0.4$ and $0.8$ respectively, with six equally likely positive states. The concentrated law is $(0.4,0.3,0.1,0.07,0.05,0.04,0.04)$.

The two-type grid uses equal type sizes, $n\in\{8,16,24\}$, $b\in\{2,3\}$, $c=(1,c_2)$ with $c_2\in\{1,4,10\}$, $w=(1,w_2)$ with $w_2\in\{1,2\}$, and $d\in\{1,W/4,W/2,W\}$. Type one has the diffuse law and type two has law $(0.7,0.15,0.05,0.04,0.03,0.02,0.01)$. The minimum-read comparison first minimizes expected reads, and then expected reading cost among attaining actions. All expected costs are finite sums evaluated in double precision. Across the tested report allocations, the solver's single-type specialization produces the same values
as the identical-source formulas, up to numerical precision.

To examine whether minimizing reads also controls cost, we consider a system in which an expensive source can resolve several units of shortage. We use a 24-task instance with equal type sizes, workloads $(1,2)$, reading costs $(1,10)$, shortage nine, and $b=3$. We compare minimum-cost access with minimum-read access, breaking ties in the latter by expected reading cost. We evaluate minimum-read access both with the joint reports held fixed and with reports reoptimized for its expected reading cost. Every comparison sums over the same source law.

\begin{table}[!htbp]
\centering
\caption{Costs under alternative access objectives. The common initial inspection cost is 132.}
\label{tab:reading-objective}
{\footnotesize\renewcommand{\arraystretch}{0.85}
\begin{tabular}{lrr}
\toprule
Reporting and access design & Recovery cost & Total cost\\
\midrule
Joint minimum-cost design & 5.768 & 137.768\\
Reports reoptimized for minimum-read access & 11.599 & 143.599\\
Joint reports with minimum-read access & 13.419 & 145.419\\
\bottomrule
\end{tabular}
}
\end{table}

Table~\ref{tab:reading-objective} shows that reoptimizing reports improves minimum-read access. Joint cost design further reduces recovery cost by 50.27\%. Including initial inspections gives a 4.06\% reduction in total access cost. Thus the reading objective changes both the preferred reports and the cost of completing the planning episode.

\subsection{Service-Assignment Experiments in Sections~\ref{subsec:implemented} and~\ref{subsec:review-experiment}}
\label{ec:llm-experiments}

For identical workloads, we use four independent states for each
$n \in \{4, 6, 8\}$ and $d \in \{1, n/2\}$, with binary reports and source law
$(0.4, 0.3, 0.15, 0.15)$. The 24 states include 21 feasible cases. Three prespecified templates render current availability and mark expired offers as unavailable. Joint reporting optimizes $h^*$ before sampling and assigns status roles to the last $h^*$ indices.

For unequal workloads, we use six independent states in each of two settings. Both have equal type sizes, binary reports, workloads $(1,2)$, and source laws $(0.4,0.3,0.15,0.15)$ and $(0.7,0.15,0.1,0.05)$. The settings are $(n,d,c_2)=(6,3,10)$ and $(8,2,4)$, with $c_1=1$. Exact prior costs determine these parameters before sampling. Ten of the 12 instances are feasible. Roles are optimized before sampling, with status roles assigned to the last selected indices within each type.

All agent experiments use a language model-based coding assistant for the planning roles.
Each agent starts with a fresh context and receives no access algorithm or
corrective feedback. One interrupted coordinator resumes in a fresh context
with its complete observable transcript. Source-access responses contain only the requested records. Repeated reads count toward cost, and malformed outputs count as
failures.

The review comparison uses cases 1, 4, 7, 10, 13, 16, 19, and 22, selected before execution. Reviewers with source access can request up to $n$ distinct full records.

An independent checker verifies report meanings, returned records, and complete final allocations or infeasibility responses. For identical workloads, paired cost intervals are obtained by bootstrap resampling of the four cases within each setting. The 95\% intervals for savings relative to all-location and all-status reports are $[0.208,0.750]$ and $[0.083,1.083]$, respectively. For each of the eight observed report vectors, we use the stated source law to calculate the minimum conditional error of a decision made without further reading. The reported zero-read risk is the arithmetic average of these eight conditional errors. Sample correctness records whether the selected action succeeds on each realized state, so the two measures can differ.

For unequal workloads, a separate recursion stores each source's possible states, checks final actions on every compatible state, and evaluates observed requests under their assigned objective. A batch commits to every requested read before another decision. Its cost, averaged over all returned states, is compared with optimal adaptive access. Initial inspection of all sources costs 26.5 on average. Exact prior costs average the optimal algorithm over all source states.

\end{document}